\documentclass[a4paper,reqno]{amsart}
\usepackage{amsmath,amssymb,mathrsfs}
\usepackage{txfonts,bm,pifont}
\usepackage{cite}
\usepackage[dvipdfm,bookmarks=true,bookmarksnumbered=true,colorlinks=true]{hyperref}
\usepackage{comment}
\usepackage[dvipdfmx]{graphicx}
\PassOptionsToPackage{svgnames}{xcolor}
\usepackage{tikz}
\usetikzlibrary{calc}
\usetikzlibrary{decorations.markings,decorations.pathmorphing}

\newtheorem{theorem}{Theorem}[section]

\newtheorem{lemma}[theorem]{Lemma}
\newtheorem{remark}[theorem]{Remark}

\numberwithin{equation}{section}
\newcommand{\wutilde}[1]{\vrule depth 0pt width 0pt%
{\raise0.8pt\hbox{$\smash{{\mathop{#1} \limits_{\displaystyle\widetilde{}}}}$}}}

\newcommand{\hT}{\hat{T}}

\newcommand{\al}{\alpha}
\newcommand{\be}{\beta}
\newcommand{\de}{\delta}
\newcommand{\ga}{\gamma}
\newcommand{\si}{\sigma}
\newcommand{\ze}{\zeta}
\newcommand{\la}{\lambda}

\newcommand{\ka}{\kappa}

\newcommand{\bbZ}{\mathbb{Z}}

\newcommand{\ii}{{\rm i}}
\newcommand{\ee}{{\rm e}}

\newcommand{\tW}{\widetilde{W}}
\newcommand{\DIV}{D_4^{(1)}}

\newcommand{\oc}[1]{{#1}^{\vee}}
\newcommand{\br}[1]{{\langle{#1}\rangle}}
\makeatletter
\long\def\@makecaption#1#2{
 \vskip 10pt
 \setbox\@tempboxa\hbox{#1. #2}
 \ifdim \wd\@tempboxa >\hsize #1. #2\par \else \hbox
to\hsize{\hfil\box\@tempboxa\hfil}
 \fi}
\makeatother
\begin{document}
\title[Geometric description of discrete power function associated with P$_{\rm VI}$]{Geometric description of discrete power function associated with the sixth Painlev\'e equation}
\author{Nalini Joshi}
\address{School of Mathematics and Statistics, The University of Sydney, New South Wales 2006, Australia.}
\email{nalini.joshi@sydney.edu.au}
\author{Kenji Kajiwara}
\address{Institute of Mathematics for Industry, Kyushu University, 744 Motooka, Fukuoka 819-0395, Japan.}
\email{kaji@imi.kyushu-u.ac.jp}
\author{Tetsu Masuda}
\address{Department of Physics and Mathematics, Aoyama Gakuin University, Sagamihara, Kanagawa 252-5258, Japan.}
\email{masuda@gem.aoyama.ac.jp}
\author{Nobutaka Nakazono}
\address{Department of Physics and Mathematics, Aoyama Gakuin University, Sagamihara, Kanagawa 252-5258, Japan.}
\email{nobua.n1222@gmail.com}
\author{Yang Shi}
\address{School of Mathematics and Statistics, The University of Sydney, New South Wales 2006, Australia.}
\email{yshi7200@gmail.com}
\begin{abstract}
In this paper, we consider the discrete power function associated with the sixth Painlev\'e equation. 
This function is a special solution of the so-called cross-ratio equation with a similarity constraint. 
We show in this paper that this system is embedded in a cubic lattice with $\tW(3A_1^{(1)})$ symmetry. 
By constructing the action of $\tW(3A_1^{(1)})$ as a subgroup of $\tW(\DIV)$, 
i.e., the symmetry group of P$_{\rm VI}$, we show how to relate $\tW(\DIV)$ to the symmetry group of the lattice. 
Moreover, by using translations in $\tW(3A_1^{(1)})$, we explain the odd-even structure appearing in previously known explicit formulas in terms of the $\tau$ function.
\end{abstract}
\maketitle
\setcounter{tocdepth}{1}

\section{Introduction}
The cross-ratio equation is a key discrete integrable equation, which connects discrete integrable systems and discrete differential geometry. 
One of its powerful roles is to provide a discrete analogue of the Cauchy-Riemann relation, thereby leading to a discrete theory of complex analytic functions.  
In this paper, we focus on the following system of partial difference equations:
\begin{align}
 &\dfrac{(z_{n,m}-z_{n+1,m})(z_{n+1,m+1}-z_{n,m+1})}{(z_{n+1,m}-z_{n+1,m+1})(z_{n,m+1}-z_{n,m})} = \frac{1}{x}, 
 \label{eqn:intro_cross-ratio}\\
 &(a_0+a_2+a_4)\,z_{n,m} = (n-a_2)\dfrac{(z_{n+1,m}-z_{n,m})(z_{n,m}-z_{n-1,m})}{z_{n+1,m}-z_{n-1,m}}\notag\\
 &\hspace{8em}+(m-a_1-a_2-a_4)\dfrac{(z_{n,m+1}-z_{n,m})(z_{n,m}-z_{n,m-1})}{z_{n,m+1}-z_{n,m-1}},\label{eqn:intro_similarity}
\end{align}
where $(n, m)\in\bbZ^2$ are independent variables, $z$ is a dependent variable and $x$, $a_i$, $i=0,1, 2, 4$, are parameters.  
Equation \eqref{eqn:intro_cross-ratio} is a special case of the cross-ratio equation and Equation \eqref{eqn:intro_similarity} is its similarity constraint \cite{NRGO2001:MR1819383}.  
The purpose of this paper is to provide a full geometric description of the system of equations \eqref{eqn:intro_cross-ratio}, \eqref{eqn:intro_similarity} and its expression in terms of $\tau$-variables.

Nijhoff {\it et al.} \cite{NRGO2001:MR1819383} showed that Equations \eqref{eqn:intro_cross-ratio} and \eqref{eqn:intro_similarity} can be regarded as a part of the B\"acklund transformations of the Painlev\'e VI equation (P$_{\rm VI}$).  
In this context, $x$ is identified as the independent variable of P$_{\rm VI}$, and 
$a_i$, $i=0,\dots,4$ ($a_3$ does not appear in the equations) are parameters corresponding to the simple roots of the affine root system of type $D_4^{(1)}$ that relates to the symmetry group of P$_{\rm VI}$
\cite{MasudaT2004:MR2097156,MasudaT2003:MR1996296,KNY2015:arXiv150908186K,OkamotoK1987:MR916698}.
However, the complete characterization of Equations \eqref{eqn:intro_cross-ratio} and \eqref{eqn:intro_similarity} in relation to the B\"acklund transformations of P$_{\rm VI}$ remains unclear.

The solution of Equations \eqref{eqn:intro_cross-ratio} and \eqref{eqn:intro_similarity} for $(n,m)\in\bbZ^2_{+}$ with the initial condition,
\begin{equation}\label{eqn:intro_initial0} 
 z_{0,0}=0,\quad z_{1,0}=1,\quad z_{0,1}={\rm e}^{r\pi \ii},
\end{equation}
where $\ii=\sqrt{-1}$ and $r=a_0+a_2+a_4$, gives rise to a discrete power function
\cite{AB2000:MR1747617,BobenkoA1996:MR1705222,BP1999:MR1676682} 
when the parameters take special values:
\begin{equation}
 x=-1,\quad 
 a_2=0,\quad
 a_1+a_2+a_4=0. 
\end{equation}
(There is more than one definition of what constitutes a discrete power function. See \cite{BooleG1860:treatise} for classical definitions.)  
Note that each of the ratios on the right hand side of \eqref{eqn:intro_similarity} is proportional to a harmonic mean of the forward and backward differences of $z_{n,m}$, so that the continuum limit of Equation \eqref{eqn:intro_similarity} is
\begin{equation}
 Z\,\frac{{\rm d} F}{{\rm d}\,Z}=X\,\frac{\partial F}{\partial X}+Y\,\frac{\partial F}{\partial Y}=2rF,
\end{equation}
which is an equation satisfied by the power function $F(Z)=Z^{2r}$, where $Z=X+\ii\,Y\in \mathbb C$.
Equations \eqref{eqn:intro_cross-ratio} and \eqref{eqn:intro_similarity} also appear as conditions for consistency of quadrilaterals on surfaces and enable a definition of discrete conformality \cite{BP1999:MR1676682,BobenkoA2008:MR2467378}.  
It was shown in \cite{AgafonovSI2005:MR2237186} that the special cases $z_{n, 0}$, $z_{n, 1}$ are given respectively by the Gamma function and Gauss hypergeometric functions.  
By applying the B\"acklund transformations, \cite{AHKM2014:MR3221837,HKM2011:MR2788707} found explicit formulas for the discrete power function in terms of the hypergeometric $\tau$ functions of P$_{\rm VI}$ as stated in Theorem \ref{theorem:HKM} below.

\begin{theorem}[\cite{AHKM2014:MR3221837,HKM2011:MR2788707}]\label{theorem:HKM}
Define the function $\tau_{\nu}(a,b,c;t)\,(c\notin\bbZ,\,\nu\in\bbZ_+)$ by 
\begin{equation}
\tau_{\nu}(a,b,c;x)=\left\{
\begin{array}{cc}
\hspace{-0.5em}\displaystyle\det\left(\varphi(a+i-1,b+j-1,c;x)\right)_{1\le i,j\le \nu}
&(\nu>0),\\
\displaystyle 1& (\nu=0),
\end{array}\right.
\label{tau}
\end{equation}
where
\begin{align}
 \varphi(a,b,c;x)=&C_0\dfrac{\Gamma(a)\Gamma(b)}{\Gamma(c)}F(a,b,c;x)\notag\\
 &+C_1\dfrac{\Gamma(a-c+1)\Gamma(b-c+1)}{\Gamma(2-c)}t^{1-c}F(a-c+1,b-c+1,2-c;x).
\end{align}
Here, $F(a,b,c;x)$ is the Gauss hypergeometric function, $\Gamma(x)$ is the Gamma function, and $C_0$ and $C_1$ are arbitrary constants. 
Then, for $(n,m)\in\bbZ^2_{+}$, $ z_{n,m}$ satisfying Equations \eqref{eqn:intro_cross-ratio} and \eqref{eqn:intro_similarity} with 
\begin{equation}
 a_2=0,\quad
 a_1+a_2+a_4=0,
\end{equation}
and the initial condition
\begin{equation}\label{eqn:initial} 
z_{0,0}=0,\quad z_{1,0}=C_0,\quad z_{0,1}=C_1x^r,
\end{equation}
where $r=a_0+a_2+a_4$, is given as follows.
\begin{enumerate}
\item Case where $n\le m$. 
\begin{equation}\label{eqn:oddeven1}
z_{n,m}=
\begin{cases}
C_1x^{r-n}N\dfrac{(r+1)_{N-1}}{(-r+1)_N}
\dfrac{\tau_{n}(-N,-r-N+1,-r;x)}{\tau_{n}(-N+1,-r-N+2,-r+2;x)}, &\text{$n+m$ is even},\\[4mm]
C_1x^{r-n}\dfrac{(r+1)_{M-1}}{(-r+1)_{M-1}}
\dfrac{\tau_{n}(-M+1,-r-M+1,-r;x)}{\tau_{n}(-M+2,-r-M+2,-r+2;x)},&\text{$n+m$ is odd}. 
\end{cases}
\end{equation}
\item Case where $n\ge m$. 
\begin{equation}\label{eqn:oddeven2}
z_{n,m}=
\begin{cases}
C_0N\dfrac{(r+1)_{N-1}}{(-r+1)_N}
\dfrac{\tau_{m}(-N+2,-r-N+1,-r+2;x)}{\tau_{m}(-N+1,-r-N+2,-r+2;x)},&\text{$n+m$ is even}, \\[4mm] 
C_0\dfrac{(r+1)_{M-1}}{(-r+1)_{M-1}}
\dfrac{\tau_{m}(-M+2,-r-M+1,-r+1;x)}{\tau_{m}(-M+1,-r-M+2,-r+1;x)}, &\text{$n+m$ is odd}.
\end{cases}
\end{equation}
\end{enumerate}
Here, $N=\frac{n+m}{2}$, $M=\frac{n+m+1}{2}$ and $(u)_j=u(u+1)\cdots(u+j-1)$ is the Pochhammer symbol. 
\end{theorem}

We remark that in the case of generic values of parameters $a_i$, the explicit expressions of $z_{n,m}$ in terms of $\tau$-variables are given in \cite{HKM2011:MR2788707}.  
This result was obtained by a heuristic approach and verified by the bilinear formulation of $\tau$-variables.
However, the appearance of the odd-even structure in the explicit formulas has not yet been fully understood.  
In \S\S \ref{section:geometric_description} and \ref{section:weight}, we will see that this odd-even structure can be explained by the notion of projective reduction and the Weyl group symmetry on $\tau$-variables.

In this paper, we characterize Equations \eqref{eqn:intro_cross-ratio} and \eqref{eqn:intro_similarity} in relation to the B\"acklund transformations of P$_{\rm VI}$ completely, and clarify the odd-even structure appearing in Equations \eqref{eqn:oddeven1} and \eqref{eqn:oddeven2}.  
The starting point is the appearance of the cross-ratio equation in a cubic lattice, as described in the ABS theory\cite{ABS2003:MR1962121,ABS2009:MR2503862,BollR2011:MR2846098,BollR2012:MR3010833,BollR:thesis}. 
We show how the cubic lattice is embedded in the lattice on which $\tau$-variables of P$_{\rm VI}$ are assigned. 
Four copies of the weight lattice of $D_4$ will be needed to describe the lattice of $\tau$-variables (see \S \ref{section:representations}).

In this paper, we refer to several lattices associated with affine reflection groups.  
A detailed description of (and notations for) these lattices is provided in \S \ref{subsection:background} below.
\subsection{Background and Notation}\label{subsection:background} 
For completeness, we explain the notation used to describe reflection groups, Coxeter groups and associated lattices in this paper.  
The notation is illustrated here for reflection groups arising from symmetries of cubes, but they apply to reflection groups of any dimension and type.
 
A cube is left unchanged by certain reflections across hyperplanes. 
All such reflections can be expressed in terms of basic operations on vectors, which form a group, denoted by $B_3$.  
(See \cite{book_HumphreysJE1992:Reflection}.)  
In higher dimensions, we would have the reflection group $B_n$, $n\ge 1$, where the subscript refers to an integer number of dimensions. 
The group is generated by simple reflections, which we denote $s_j$, $j=1,\dots, n$. 
The simple roots associated with these reflections are denoted $\alpha_j$, $j=1,\dots, n$, while corresponding co-roots are denoted $\oc\alpha_j$.

Repeated translation of a fundamental cube leads to a space-filling cubic lattice, on which Equations \eqref{eqn:intro_cross-ratio} and \eqref{eqn:intro_similarity} are iterated.  
Reflections across hyperplanes in the cubic lattice form an affine reflection group, which is denoted by $B_n^{(1)}$. 
We will refer also to the root lattice denoted by $Q(B_n^{(1)})=\oplus_{j=0}^{n} \mathbb Z\alpha_j$. 
When interpreted as a Coxeter group, the reflection group is denoted by $W(B_n^{(1)})$.

In a Dynkin diagram, the simple reflections $s_j$ are represented by nodes.  
Obviously, such a diagram remains unchanged when the nodes are permuted. 
Such an operation is called a diagram automorphism. 
Let $\Omega$ be the group of such automorphisms. 
Incorporating these automorphisms into the space of allowable operations, we get a lift of the Coxeter group $W(B_n^{(1)})$ to an {\it extended} Coxeter group denoted $\widetilde W(B_n^{(1)})=W(B_n^{(1)})\rtimes \Omega$.  
(See \cite{NY1998:MR1666847}.)

The classification of finite reflection groups leads to several fundamental types.  
We will encounter $A_n$, $B_n$, and $D_n$ in this paper, along with their corresponding lattices, affine Coxeter groups and extended Coxeter groups. 
Where direct sums of the same type of groups occur, we use an integer coefficient as an abbreviation, e.g., $A_1\oplus A_1\oplus A_1\oplus A_1=:4 A_1$.
For conciseness, we replace $\oplus$ by $+$, e.g., $D_4\oplus A_1\oplus A_1=D_4+2 A_1$.

\subsection{Main result and plan of the paper}
The main contribution of this paper, given in Theorem \ref{theorem:main}, is to clarify the relation of the discrete power function to B\"acklund transformations of the sixth Painlev\'e equation P$_{\rm VI}$ by including the system of defining difference equations into a lattice with $\tW(3 A_1^{(1)})$ symmetry. 
In general, a cubic lattice has $W(B_3^{(1)})$ as a symmetry group, but we chose to consider the symmetry group as $\tW(3 A_1^{(1)})$ by taking different selections of reflection hyperplanes.  
This idea allows us to realize $\tW(3A_1^{(1)})$ as a subgroup of $\tW(\DIV)$, which is the symmetry group of P$_{\rm VI}$.
To achieve this realization, we use the combination of the cross-ratio equation together with its similarity relation, which is associated with the symmetry of the cubic lattice.

The plan of this paper is as follows.  
In \S \ref{section:cubic_lattice} we introduce a multidimensionally consistent cubic lattice with the symmetry $\tW(3A_1^{(1)})$ on which the cross-ratio equation is naturally defined.  
In \S \ref{section:representations} we give actions of $\tW(\DIV)$ on the parameters and the $\tau$-variables.
Based on this material, we construct a realization of $\tW(3A_1^{(1)})$ in $\tW(D_4^{(1)})$ and derive Equations \eqref{eqn:intro_cross-ratio} and \eqref{eqn:intro_similarity} in \S \ref{section:geometric_description}.  
We then explain in \S \ref{section:weight} the odd-even structures appearing in the explicit formula in \cite{AHKM2014:MR3221837,HKM2011:MR2788707}.  
Concluding remarks will be given in \S \ref{ConcludingRemarks}.

\section{Cubic Lattice}\label{section:cubic_lattice}
In this section, we describe a multidimensionally consistent cubic lattice in terms of a symmetry structure, which will turn out to be related to the symmetry group of P$_{\rm VI}$.

While each cube inside the lattice is composed of 6 component faces, we will regard these as composed of two triple faces, i.e., a pair of 3 faces around a common vertex.  
We will see in \S \ref{section:geometric_description} how this perspective is embedded in the symmetry-group-lattice of P$_{\rm VI}$.  
The symmetry group $\tW(3 A_1^{(1)})$ is suggested by the similarity equation:
\begin{equation}\label{eqn:similarity}
\begin{split}
  \ze_0\, z_{l_1,l_2}
 =&(l_1-\be_1)\cfrac{(z_{l_1+1,l_2}-z_{l_1,l_2})(z_{l_1,l_2}-z_{l_1-1,l_2})}{z_{l_1+1,l_2}-z_{l_1-1,l_2}}\\
 &+(l_2-\ga_1)\cfrac{(z_{l_1,l_2+1}-z_{l_1,l_2})(z_{l_1,l_2}-z_{l_1,l_2-1})}{z_{l_1,l_2+1}-z_{l_1,l_2-1}}.
\end{split}
\end{equation}
Here, 
\begin{equation}
 \be_1=a_2,\quad
 \ga_1=a_1+a_2+a_4,\quad
 \ze_0=a_0+a_2+a_4.
\end{equation}
Indeed, the similarity equation has three parameters $\be_1$, $\ga_1$ and $\ze_0$ and two directions $\rho_1$ and $\rho_2$, defined by $z_{l_1,l_2}={\rho_1}^{l_1}{\rho_2}^{l_2}(z_{0,0})$, which shift the parameters as follows:
\begin{equation}
 \rho_1:(\be_1,\ga_1,\ze_0)\mapsto(\be_1-1,\ga_1,\ze_0),\quad
 \rho_2:(\be_1,\ga_1,\ze_0)\mapsto(\be_1,\ga_1-1,\ze_0).
\end{equation}
It is therefore natural to expect that there exists a B\"acklund transformation of the similarity equation, which shifts the parameters as follows:
\begin{equation}
 \rho_0:(\be_1,\ga_1,\ze_0)\mapsto(\be_1,\ga_1,\ze_0+1).
\end{equation}
This implies that the three directions $\rho_0, \rho_1, \rho_2$ are translations in $\tW(3A_1^{(1)})$ and $\{\be_1,\ga_1,\ze_0\}$ are the parameters associated with the simple roots (or, coroots) of type $3A_1^{(1)}$.  
Consequently, we take the symmetry group of the cubic lattice to be $\tW(3 A_1^{(1)})$.

Consider the cubic lattice constructed by the directions $\rho_i$, $i=0,1,2$.  
We place the following equations on each respective face of a triple of faces associated with each cube (referred to as the face equations):
\begin{subequations}\label{eqns:3D_lattice_kappa}
\begin{align}
 &\cfrac{(u_{l_1,l_2,l_0}+u_{l_1+1,l_2,l_0})(u_{l_1+1,l_2+1,l_0}+u_{l_1,l_2+1,l_0})}
 {(u_{l_1+1,l_2,l_0}+u_{l_1+1,l_2+1,l_0})(u_{l_1,l_2+1,l_0}+u_{l_1,l_2,l_0})}=\cfrac{\ka^{(1)}_{l_1}}{\ka^{(2)}_{l_2}}\,,\label{eqn:Q1H1_1}\\
 &\left(\frac{1}{u_{l_1+1,l_2,l_0}}+\frac{1}{u_{l_1,l_2,l_0}}\right) 
 (u_{l_1+1,l_2,l_0+1}+u_{l_1,l_2,l_0+1})
 =-\ka^{(1)}_{l_1}\ka^{(3)}_{l_0},\label{eqn:Q1H1_2}\\
 &\left(\frac{1}{u_{l_1,l_2+1,l_0}}+\frac{1}{u_{l_1,l_2,l_0}}\right)
 (u_{l_1,l_2+1,l_0+1}+u_{l_1,l_2,l_0+1})
 =-\ka^{(2)}_{l_2}\ka^{(3)}_{l_0},\label{eqn:Q1H1_3}
\end{align}
\end{subequations}
where $\{\dots,\ka^{(i)}_{-1},\ka^{(i)}_0,\ka^{(i)}_1,\dots\}_{i=1,2,3}$ are parameters and
\begin{equation}
 u_{l_1,l_2,l_0}={\rho_1}^{l_1}{\rho_2}^{l_2}{\rho_0}^{l_0}(u_{0,0,0}).
\end{equation}
It is well known that this system of equations is multidimensionally consistent in the cubic lattice.  
Indeed, Equation \eqref{eqn:Q1H1_1} and Equations \eqref{eqn:Q1H1_2} and \eqref{eqn:Q1H1_3} are called Q1 and H1 respectively in the ABS list 
(see Figure \ref{fig:3Dlattice1})\cite{ABS2003:MR1962121,ABS2009:MR2503862,BollR2011:MR2846098,BollR2012:MR3010833,BollR:thesis}.
\begin{figure}[t]
\begin{center}
\includegraphics[width=0.6\textwidth]{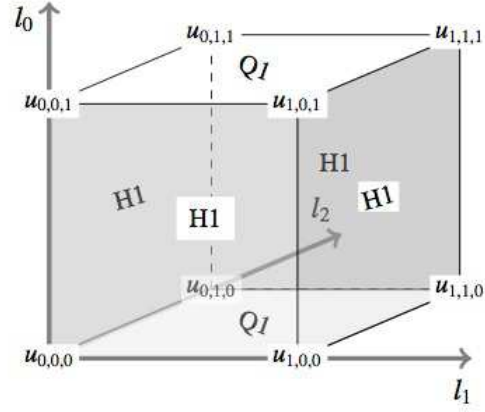}
\end{center}
\caption{The multidimensionally consistent cubic lattice. The face equations of bottom and top are given by Q1 and those of sides are given by H1.}
\label{fig:3Dlattice1}
\end{figure}

Using the geometric theory of P$_{\rm VI}$, we can realize the action of $\tW(3A_1^{(1)})$ associated with these equations (see \S \ref{section:geometric_description}).

\section{Actions of affine Weyl group $\tW(\DIV)$}\label{section:representations}
In this section, we provide the basic ingredients needed to describe actions of the symmetry group of P$_{\rm VI}$, i.e., $\tW(\DIV)$.  The explicit formulation of the $\tau$ function of P$_{\rm VI}$ in terms of the weight lattice is provided here; to our knowledge, such an explicit formulation has not appeared in the literature.  
More precisely, four copies of the weight lattice of $D_4$ will be needed to define the lattice of $\tau$-variables.  
We will consider a root system of type $D_4+2A_1$ instead of $D_4$ to give a full description.

\subsection{Linear action}
In this subsection, we define the transformation group $\tW(\DIV)$ and describe its linear actions on the weight lattice.

We consider the following $\bbZ$-modules:
\begin{align}
 &\oc Q=\bbZ\oc\al_0+\bbZ\oc\al_1+\bbZ\oc\al_2+\bbZ\oc\al_3+\bbZ\oc\al_4+\bbZ\oc\al_5+\bbZ\oc\al_6,\\
 &P=\bbZ h_0+\bbZ h_1+\bbZ h_2+\bbZ h_3+\bbZ h_4+\bbZ h_5+\bbZ h_6,
\end{align}
with the bilinear form $\langle\cdot,\cdot\rangle:\oc Q\times P\to\bbZ$ defined by
\begin{equation}\label{eqn:def_dual}
 \langle\oc\al_i,h_j\rangle=\de_{ij},\quad 0\leq i, j \leq 6.
\end{equation}
The weight lattice of type $\DIV$ is spanned by $h_i$, $i=1,3,4$, and $h_2-h_0$.
The generators $\{\oc\al_0,\dots,\oc\al_4\}$ and $\{\oc\al_5,\oc\al_6\}$ are identified with the coroot of type $\DIV$ and type $2A_1$, respectively.
Let us define the roots of type $D_4^{(1)}+2A_1$: $\{\al_0,\dots,\al_6\}$ which satisfy
\begin{equation}
 \langle\oc\al_i,\al_j\rangle
 =\begin{cases}
 A_{ij},& 0\leq i,j\leq 4\\
 2,&(i,j)=(5,5),~(6,6)\\
 0, &\text{otherwise},
 \end{cases}
\end{equation}
by the following:
\begin{equation}\label{eqn:def_alpha}
 \begin{pmatrix}
 \al_0\\ \al_1\\ \al_2\\ \al_3\\ \al_4
 \end{pmatrix} 
 =(A_{ij})_{i, j=0}^4
 \begin{pmatrix}
 h_0\\ h_1\\ h_2\\ h_3\\ h_4
 \end{pmatrix},\quad
 \al_5=2h_5,\quad
 \al_6=2h_6.
\end{equation}
Here, $(A_{ij})_{i, j=0}^4$ is the Generalized Cartan matrix of type $D_4^{(1)}$:
\begin{equation}\label{eqn:cartan_D4}
 (A_{ij})_{i, j=0}^4
 =\begin{pmatrix}
 2&0&-1&0&0\\
 0&2&-1&0&0\\
 -1&-1&2&-1&-1\\
 0&0&-1&2&0\\
 0&0&-1&0&2
 \end{pmatrix}.
\end{equation}
We note that 
\begin{equation}\label{eqn:cond_root_al}
 \al_0+\al_1+2\al_2+\al_3+\al_4=0.
\end{equation}

We define the transformations $s_i$, $i=0,\dots,4$, which are reflections for the roots
$\{\al_0,\dots,\al_4\}$, by
\begin{equation}\label{eqns:s_weight}
 s_i(\la)=\la-\langle\oc\al_i,\la\rangle\al_i,\quad i=0,\dots,4,\quad \la\in P,
\end{equation}
which give
\begin{equation}\label{eqn:s_h}
 s_i:h_i\mapsto h_2-h_i,\quad i=0,1,3,4,\quad
 s_2: h_2\mapsto h_0+h_1-h_2+h_3+h_4.
\end{equation}
Note that the ``$h_i$"s which are not explicitly shown in Equation \eqref{eqn:s_h} remain unchanged.
Also, define the transformations $\si_i$, $i=1,2,3$, which are the automorphisms of the Dynkin
diagram of type $\DIV$ (see Figure \ref{fig:D4_Dynkin}) and the reflections for the simple roots
$\al_5$ and $\al_6$, by
\begin{subequations}\label{eqns:sigma_weight}
\begin{align}
 &\si_1:
 \left(\begin{array}{l}
 h_0,~h_1,~h_2,~h_3,~h_4\\
 h_5,~h_6
 \end{array}\right)
 \mapsto
 \left(\begin{array}{l}
 h_1+h_5,~h_0+h_5,~h_2+2h_5,~h_4+h_5,~h_3+h_5\\
 -h_5,~h_6
 \end{array}\right),\label{eqn:sigma1_weight}\\
 &\si_2:
 \left(\begin{array}{l}
 h_0,~h_1,~h_2,~h_3,~h_4\\
 h_5,~h_6
 \end{array}\right)
 \mapsto
 \left(\begin{array}{l}
 h_3+h_6,~h_4+h_6,~h_2+2h_6,~h_0+h_6,~h_1+h_6\\
 h_5,~-h_6
 \end{array}\right),\\
 &\si_3=\si_1\si_2.
\end{align}
\end{subequations}
From definitions \eqref{eqn:def_dual}, \eqref{eqn:def_alpha}, \eqref{eqns:s_weight} and
\eqref{eqns:sigma_weight}, we can compute actions on the simple roots $\al_i$, $i=0,\dots,6$:
\begin{subequations}\label{eqns:action_D4_alpha}
\begin{align}
 &s_i:(\al_i,\al_2)\mapsto(-\al_i,\al_2+\al_i),\quad i=0,1,3,4,\\
 &s_2:(\al_0,\al_1,\al_2,\al_3,\al_4)\mapsto
 (\al_0+\al_2,\al_1+\al_2,-\al_2,\al_3+\al_2,\al_4+\al_2),\\
 &\si_1:(\al_0,\al_1,\al_3,\al_4,\al_5)\mapsto
 (\al_1,\al_0,\al_4,\al_3,-\al_5),\\
 &\si_2:(\al_0,\al_1,\al_3,\al_4,\al_6)\mapsto
 (\al_3,\al_4,\al_0,\al_1,-\al_6),\\
 &\si_3:(\al_0,\al_1,\al_3,\al_4,\al_5,\al_6)\mapsto
 (\al_4,\al_3,\al_1,\al_0,-\al_5,-\al_6).
\end{align}
\end{subequations}
Under the linear actions on the weight lattice \eqref{eqns:s_weight} and \eqref{eqns:sigma_weight},
$\tW(\DIV)=\langle s_0,\dots,s_4,\si_1,\si_2\rangle$ forms an extended affine Weyl group of type
$\DIV$.  Indeed, the following fundamental relations hold:
\begin{subequations}\label{eqns:fundamental_rels_D4}
\begin{align}
 &(s_is_j)^{m_{ij}}=1,\quad 0\leq i,j\leq 4,\qquad
 {\si_i}^2=1,\quad i=1,2,\\
 &\si_1s_{\{0,1,2,3,4\}}=s_{\{1,0,2,4,3\}}\si_1,\quad
 \si_2s_{\{0,1,2,3,4\}}=s_{\{3,4,2,0,1\}}\si_2,\quad
 \si_1\si_2=\si_2\si_1,\label{eqns:fundamental_rels_sigma}
\end{align}
\end{subequations}
where 
\begin{equation}
 m_{ij}=
\begin{cases}
 1,& i=j\\
 3, &i=2,~j\ne 2\quad \text{or}\quad i\neq2,~j=2\\
 2, &\text{otherwise}.
\end{cases}
\end{equation}

\begin{figure}[t]
\begin{center}
\includegraphics[width=0.4\textwidth]{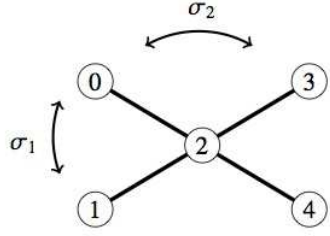}
\end{center}
\caption{Dynkin diagram of type $D_4^{(1)}$.}
\label{fig:D4_Dynkin}
\end{figure}

\begin{remark}\label{rem:bilinear_w}
We can also define the action of $\tW(\DIV)$ on the coroot lattice $\oc Q$ by replacing $\al_i$ with $\oc\al_i$ in the actions \eqref{eqns:action_D4_alpha}.  
Then, the transformations in $\tW(\DIV)$ preserve the form $\langle\cdot,\cdot\rangle$, that is, the following holds:
\begin{equation}
 \langle\ga,\la\rangle=\langle w(\ga),w(\la)\rangle,
\end{equation}
for arbitrary $w\in\tW(\DIV)$, $\ga\in\oc Q$ and $\la\in P$.
\end{remark}

\subsection{Action on $\tau$-variables}\label{subsection:birational_D4}
It is well known that we can extend the linear action of $\widetilde{W}(D_4^{(1)})$ to the action of the $\tau$ function of the Painlev\'e VI equation\cite{MasudaT2004:MR2097156,MasudaT2003:MR1996296}.  
In this section, we define the $\tau$ function on the weight lattice and give the action of $\widetilde{W}(D_4^{(1)})$ on it.

Let $M_k$, $k=0,\dots,4$, be the orbits of $h_i$, $i=0,\dots,4$, defined by  
\begin{equation}\label{eqn:def_M_i}
 M_k=\left\{w(h_k)\,\left|\, w\in\tW(\DIV)\right\}\right.,\quad k=0,\dots,4,
\end{equation}
and $M$ be their disjoint union defined by 
\begin{equation}\label{eqn:def_M}
 M=\sqcup_{k=0}^4 M_k.
\end{equation}
Each of $M_k$ is decomposed to the orbit of $W(\DIV)$ by 
\begin{equation}
 M_k=M_k^{(00)}\cup M_k^{(10)}\cup M_k^{(01)}\cup M_k^{(11)},
\end{equation}
where 
\begin{align}
 &M_k^{(ij)}=\left\{\la\in M_k\,\left|\,\br{\al_5^{\vee},\la}=i,~\br{\al_6^{\vee},\la}=j\right\}\right.,
 \quad k=0,1,3,4,\\
 &M_2^{(ij)}=\left\{\la\in M_2\,\left|\,\br{\al_5^{\vee},\la}=2i,~\br{\al_6^{\vee},\la}=2j\right\}\right..
\end{align}
For fixed $k\in\{0,1,2,3,4\}$, $M_k^{(ij)}$ are transformed to each other by the action of Dynkin automorphisms $\si_i$, $i=1,2,3$. 
For example, the weight $h_0$ lies in $M_0^{(00)}$ and the orbit of $h_0$ under $\si_1$ and $\si_2$ can be calculated by Equations \eqref{eqn:def_dual} and \eqref{eqns:sigma_weight}, which give 
\begin{equation}
 \si_1(h_0)=h_1+h_5\in M_0^{(10)},\quad
 \si_2(h_0)=h_3+h_6\in M_0^{(01)}. 
\end{equation}

\begin{remark}\label{remark:Mk(lm)}
~\\[-1.2em]
\begin{enumerate}
\item[(i)] 
Given any $w\in\tW(\DIV)$, the element $w(h_k)$, $k=0,1,3,4$, can be expressed as $\sum_{i=0}^4n_ih_i+l\,h_5+m\,h_6$, where $n_i\in\bbZ$, $l,m\in\{0,1\}$. 
By computing bilinear forms with $\oc\al_5$ and $\oc\al_6$, we see that $w(h_k)$ lies in $M_k^{(lm)}$.  
\item[(ii)] 
From the relations \eqref{eqns:fundamental_rels_sigma}, we can express any element of $\tW(\DIV)$ as one of
\begin{equation*}
 w_1,\quad
 \si_1w_2,\quad
 \si_2w_3,\quad
 \si_1\si_2w_4,
\end{equation*}
where $w_i\in W(\DIV)$, $i=1,\dots,4$. 
Note that the action of $w_i$, $i=1,\dots,4$, on $\sum_{i=0}^4 n_ih_i+l\,h_5+m\,h_6$, does not change $l$ and $m$ but $\si_i$, $i=1,2$, does change them.  
\item[(iii)] 
Based on these observations in {\rm(i)} and {\rm(ii)}, it follows that if $w(h_k)\in M_k^{(lm)}$, where $k=0,1,3,4$ and $w\in\tW(\DIV)$, the superscript ``$(lm)$'' of $M_k^{(lm)}$ is determined by the number of iterations of $\si_1$ and $\si_2$ occurring in $w$ as:
\begin{equation}
 l=\begin{cases}
 0\quad\text{(if number of $\si_1$ is even)},\\
 1\quad\text{(if number of $\si_1$ is odd)},
 \end{cases}\quad
 m=\begin{cases}
 0\quad\text{(if number of $\si_2$ is even)},\\
 1\quad\text{(if number of $\si_2$ is odd)}.
 \end{cases}
\end{equation}
This is easily verified from the actions of $\si_i$\,, $i=1,2$, on $h_k$ given in \eqref{eqns:sigma_weight}.
\item[(iv)] 
It follows that for any element $w(h_k)=\sum_{i=0}^4 n_i h_i +l\,h_5 + m\,h_6$,  
where $k=0,1,3,4$ and $w\in\tW(\DIV)$, 
$l$ and $m$ are uniquely determined from $n_i$, $i=0,\dots,4$.
For example, suppose that $w\in\tW(D_4^{(1)})$ is decomposed as $\sigma_1 w_1$, where $w_1\in W(D_4^{(1)})$.
Then we have from \eqref{eqns:sigma_weight} that 
\begin{equation}
 w_1(h_0)=n_1 h_0+n_0 h_1+n_2 h_2+n_4 h_3+n_3 h_4+(n_0 + n_1 + 2n_2 + n_3 + n_4 - l)h_5 + m h_6.
\end{equation}
Since $w_1(h_0)$ should not include $h_5$ and $h_6$ as seen in Equation \eqref{eqn:s_h}, 
we have $l=n_0 + n_1 + 2n_2 + n_3 + n_4$ and $m=0$. 
\end{enumerate}
\end{remark}

\begin{lemma}
The following properties hold.
\begin{description}
\item[(i)] 
Lattices $M_k$ are mutually disjoint. 
\item[(ii)] 
Each of $M_k$, $k=0,1,3,4$, is isomorphic to the weight lattice of type $\DIV$.
\end{description}
\end{lemma}
\begin{proof}
First we prove the property {\bf(i)}.
Define
\begin{equation}
 \de^\vee=\oc\al_0+\oc\al_1+2\oc\al_2+\oc\al_3+\oc\al_4,
\end{equation}
which is invariant under the action of $\tW(\DIV)$.  
From Remark \ref{rem:bilinear_w}, it is obvious that $M_2$ is disjoint from the others, since $\br{\de^\vee,\la}=2$ for any $\la\in M_2$ and $\br{\de^\vee,\la}=1$ for any $\la\in M_k$, $k=0,1,3,4$. 
Let us consider the 16 lattices $M_k^{(ij)}$, $k=0,1,3,4$.  
It is easy to see that these are mutually disjoint if the superscripts $(ij)$ are different from each other.  
Then it is sufficient to show that $M_k^{(00)}$, $k=0,1,3,4$, are mutually disjoint.  
We prove it by contradiction.  
Suppose that there exists an element of $M_0^{(00)}\cap M_1^{(00)}$.  
Since we are now restricting our attention to $M_k^{(00)}$, we use only the actions of $W(\DIV)$.  
Then, there exists an element $w\in W(\DIV)$ such that $h_1=w(h_0)$ which is the reflection with respect to the root vector $h_1-h_0$.  
Therefore, it contradicts the assumption.

Next, we show the property {\bf(ii)}. 
We prove it for $M_0$, by constructing a map defined by
\begin{equation}
 p:M_0 \mapsto \bigoplus_{k=0}^4{\mathbb Z}h_k
\end{equation}
where every element of $M_0$ is given by $\sum_{i=0}^4n_ih_i+l\,h_5+m\,h_6$, with the proviso that $l$ and $m$ are determined by $\{n_i\}$, as mentioned in Remark \ref{remark:Mk(lm)}. 
We define the image of such an element under $p$ as $\sum_{i=0}^4n_ih_i$. 
It follows from this construction that $p$ is injective, because two elements in the domain of $p$ with the same image must therefore be identical. 
Since $p$ is injective and $\br{\delta^{\vee},p(\lambda)}=1$ for $\la\in M_0$, lattice $M_0$ is isomorphic to the weight lattice of type $D_4$.
Note that we have
\begin{align}
 p(M_0)&=p(M_0^{(00)})\cup p(M_0^{(10)})\cup p(M_0^{(01)})\cup p(M_0^{(11)})\notag\\
 &\cong M_0^{(00)}\cup M_1^{(00)}\cup M_3^{(00)}\cup M_4^{(00)}.
\end{align}
In a similar manner, we can prove the property {\bf(ii)} for $M_k$, $k=1,3,4$.
Therefore, we have completed the proof.
\end{proof}

We now consider the $\tau$-variables assigned on the lattice $M$, and action of $\tW(\DIV)$ on them. 
As we will show in \S \ref{section:geometric_description}, the dependent variable of the system of partial difference equations \eqref{eqn:intro_cross-ratio} and \eqref{eqn:intro_similarity} is defined by a ratio of $\tau$-variables, and the system is derived
from the action of the Weyl group on the $\tau$-variables.

First, let $\{\tau_i\}_{i=0}^4$, $\tau_2^{(\si_1)}$, $\tau_2^{(\si_2)}$, $\tau_2^{(\si_3)}$,
$\tau_2^{(s_2)}$ be the variables and $a_i$, $i=0,\dots,4$, $x$ be the complex parameters satisfying
the following conditions:
\begin{align}
 &x\tau_2+x^{1/2}(x-1)^{1/2}\tau_2^{(\si_1)}
 =-\ii x^{1/2}\tau_2^{(\si_2)}
 =\tau_2-\ii(x-1)^{1/2}\tau_2^{(\si_3)},\\
 &a_0+a_1+2a_2+a_3+a_4=1.\label{eqn:cond_para_a}
\end{align}
Therefore, the number of essential variables and parameters are 7 and 5, respectively.  
We now proceed to give the actions of $\tW(\DIV)$ on these $\tau$-variables, which were deduced in
\cite{MasudaT2004:MR2097156,MasudaT2003:MR1996296}:
\begin{equation}\label{eqns:action_D4_tau}
\begin{split}
 &s_0\,:\,
 \tau_0\mapsto u^2v^2\cfrac{\tau_2^{(\si_1)}}{\tau _0},\quad
 \tau_2^{(s_2)}\mapsto\cfrac{u^2v^2\tau_2^{(\si_1)}\tau_2^{(s_2)}-a_0u^{-2}\tau_{0134}}{{\tau_0}^2},\\
 &s_1\,:\,
 \tau_1\mapsto\cfrac{\tau_2}{\tau_1},\quad
            \tau_2^{(s_2)}\mapsto\cfrac{\tau_2^{(s_2)}\tau_2}{{\tau_1}^2},\\
 &s_2\,:\,
\left\{
\begin{array}{l}
 \tau_2\mapsto\tau_2^{(s_2)},\quad
 \tau_2^{(s_2)}\mapsto\tau_2,\quad
 \tau_2^{(\si_1)}\mapsto\cfrac{\tau_2^{(\si_1)}\tau_2^{(s_2)}+a_2u^{-4}v^{-2}\tau_{0134}}{\tau_2},\\[4mm]
 \tau_2^{(\si_2)}\mapsto\cfrac{\tau_2^{(\si_2)}\tau_2^{(s_2)}+\ii a_2u^{-4}\tau_{0134}}{\tau_2},\quad
 \tau_2^{(\si_3)}\mapsto\cfrac{\tau_2^{(\si_3)}\tau_2^{(s_2)}+\ii a_2u^{-2}v^{-2}\tau_{0134}}{\tau_2},
\end{array}
\right.
\\
 &s_3\,:\,
 \tau_3\mapsto-\ii v^2\cfrac{\tau_2^{(\si_3)}}{\tau_3},\quad
 \tau_2^{(s_2)}\mapsto\cfrac{-\ii v^2\tau_2^{(\si_3)}\tau_2^{(s_2)}-a_3u^{-2}\tau_{0134}}{{\tau_3}^2},\\
 &s_4\,:\,
 \tau_4\mapsto-\ii u^2\frac{\tau_2^{(\si_2)}}{\tau_4},\quad
 \tau_2^{(s_2)}\mapsto\cfrac{-\ii u^2\tau_2^{(\si_2)}\tau_2^{(s_2)}-a_4u^{-2}\tau_{0134}}{{\tau_4}^2},\\
 &\si_1\,:\,
\left\{
\begin{array}{l}
 \tau_0\mapsto uv\tau_1,\quad
 \tau_1\mapsto u^{-1}v^{-1} \tau_0,\quad
 \tau_3\mapsto u^{-1}v \tau_4,\quad
 \tau_4\mapsto uv^{-1}\tau_3,\\[3mm]
 \tau_2\mapsto\tau_2^{(\si_1)},\quad
 \tau_2^{(\si_1)}\mapsto\tau_2,\quad
 \tau_2^{(\si_2)}\mapsto\tau_2^{(\si_3)},\quad
 \tau_2^{(\si_3)}\mapsto\tau_2^{(\si_2)},\\[3mm]
 \tau_2^{(s_2)}\mapsto-s_2(\tau_2^{(\si_1)}),
\end{array}\right.\\[4mm]
 &\si_2\,:\,
\left\{
\begin{array}{l}
 \tau_0\mapsto\ee^{-\pi \ii/4} u \tau_3,\ \ 
 \tau_1\mapsto\ee^{\pi \ii/4} u^{-1} \tau_4,\ \  
 \tau_3\mapsto\ee^{-\pi \ii/4} u^{-1} \tau_0,\ \ 
 \tau_4\mapsto\ee^{\pi \ii/4} u \tau_1,\\[3mm]
 \tau_2\mapsto\tau_2^{(\si_2)},\quad
 \tau_2^{(\si_1)}\mapsto-\tau_2^{(\si_3)},\quad
 \tau_2^{(\si_2)}\mapsto-\tau_2,\quad
 \tau_2^{(\si_3)}\mapsto\tau_2^{(\si_1)},\\[3mm]
 \tau_2^{(s_2)}\mapsto s_2(\tau_2^{(\si_2)}),
\end{array}\right.
\end{split}
\end{equation}
where 
\begin{equation}
 u=x^{1/4},\quad
 v=(x-1)^{1/4},\quad
 \tau_{0134}=\tau_0\tau_1\tau_3\tau_4.
\end{equation}
Here, we define the action on the parameters $a_i$, $i=0,\dots,4$, by
\begin{subequations}
\begin{align}
 &s_i(a_j)=a_j-A_{ij}a_i,\quad 0\leq i,j\leq4,\\
 &\si_1a_{\{0,1,2,3,4\}}=a_{\{1,0,2,4,3\}},\quad
 \si_2a_{\{0,1,2,3,4\}}=a_{\{3,4,2,0,1\}},
\end{align}
\end{subequations}
where $A_{ij}$ is given in \eqref{eqn:cartan_D4}. 
Note that each parameter $a_i$ is associated with the root $\al_i$, $i=0,\dots,4$,
and the condition \eqref{eqn:cond_para_a} corresponds to \eqref{eqn:cond_root_al}.

The fundamental relations \eqref{eqns:fundamental_rels_D4} also hold under the actions $s_i$, $i=0,\dots,4$, on the $\tau$-variables and the parameters. 
Note that the extended part $\si_i$, $i=1,2$, are modified to become
\begin{subequations}\label{eqns:fundamental_tau_gaps}
\begin{align}
 &{\si_2}^2.
 \left(\begin{array}{l}
 \tau_0,~\tau_1,~\tau_2,~\tau_3,~\tau_4\\
 \tau_2^{(\si_1)},~\tau_2^{(\si_2)},~\tau_2^{(\si_3)},~\tau_2^{(s_2)}
 \end{array}\right)
 =\left(\begin{array}{l}
 -\ii\tau_0,~\ii\tau_1,~-\tau_2,~-\ii\tau_3,~\ii\tau_4\\
 -\tau_2^{(\si_1)},~-\tau_2^{(\si_2)},~-\tau_2^{(\si_3)},~-\tau_2^{(s_2)}
 \end{array}\right),\\
 &\si_1s_2.\left(\tau_2,\tau_2^{(\si_1)},\tau_2^{(\si_2)},\tau_2^{(\si_3)},\tau_2^{(s_2)}\right)
 =s_2\si_1.\left(-\tau_2,-\tau_2^{(\si_1)},-\tau_2^{(\si_2)},-\tau_2^{(\si_3)},-\tau_2^{(s_2)}\right),\\
 &\si_1\si_2.
 \left(\begin{array}{l}
 \tau_0,~\tau_1,~\tau_2,~\tau_3,~\tau_4\\
 \tau_2^{(\si_1)},~\tau_2^{(\si_2)},~\tau_2^{(\si_3)},~\tau_2^{(s_2)}
 \end{array}\right)
 =\si_2\si_1.\left(\begin{array}{l}
 -\ii\tau_0,~\ii\tau_1,~-\tau_2,~-\ii\tau_3,~\ii\tau_4\\
 -\tau_2^{(\si_1)},~-\tau_2^{(\si_2)},~-\tau_2^{(\si_3)},~-\tau_2^{(s_2)}
 \end{array}\right).
\end{align}
\end{subequations}

\begin{remark}
The correspondence between the notations in this paper and those in \cite{MasudaT2004:MR2097156,MasudaT2003:MR1996296} is given by
\begin{equation}
 \left(\begin{array}{l}
 a_0,a_1,a_2,a_3,a_4,x\\
 s_0,s_1,s_2,s_3,s_4,\si_1,\si_2,\si_3\\
 \tau_0,~\tau_1,~\tau_2,~\tau_3,~\tau_4\\
 \tau_2^{(\si_1)},~\tau_2^{(\si_2)},~\tau_2^{(\si_3)},~\tau_2^{(s_2)}
 \end{array}\right)
 \to
 \left(\begin{array}{l}
 \al_0,\al_1,\al_2,\al_3,\al_4,t\\
 s_0,s_1,s_2,s_3,s_4,s_5,s_6,s_7\\
 \tau_0,~\tau_1,~\tau_2,~\tau_3,~\tau_4\\
 s_5(\tau_2),~s_6(\tau_2),~s_7(\tau_2),~s_2(\tau_2)
 \end{array}\right).
\end{equation}
\end{remark}

We then define a mapping $\tau$ on $M$ by
\begin{equation}
 \tau(h_i) = \tau_i,\quad
 \tau(\si_j.h_2)=\tau_2^{(\si_j)},\quad 
\tau(s_2.h_2)=\tau_2^{(s_2)},\quad
\tau(w.h_i) = w(\tau_i),
\end{equation}
where $i=0,\dots,4$, $j=1,2,3$ and $w\in\tW(D_4^{(1)})$.

\section{Geometric description of discrete power function}\label{section:geometric_description} 
In this section, we provide our main result, which is based on the action of $\tW(D_4^{(1)})$ constructed in the previous sections.

\begin{theorem}\label{theorem:main}
~\\[-1.2em]
\begin{enumerate}
\item[{\bf I.}] 
There exist elements $\rho_i\in\tW(\DIV)$, $i=0, 1,2$, which satisfy the following properties.
\begin{description}
\item[(i)] {\rm (}Projective reduction{\rm )}
For each $i=0, 1,2$, the element ${\rho_i}^2$ is a translation of $\tW(\DIV)$, while $\rho_i$ itself is not a translation.
\item[(ii)]   {\rm (}Commutativity{\rm )}
$\rho_i\,\rho_j=\rho_j\,\rho_i$\,,~ for $i, j=0, 1,2$.
 \item[(iii)]  {\rm (}Bilinear forms{\rm )}
There exists a subspace $\bbZ\be_1\oplus\bbZ\ga_1\oplus\bbZ\ze_1$, of the parameter space $\bbZ a_1\oplus\bbZ a_2\oplus\bbZ a_3\oplus\bbZ a_4$, on which the actions become translations.
That is, the projections of $\rho_i$, $i=0,1,2$ on the sub-space satisfy the following properties:  
\begin{equation}
 (\be_1|\ga_1)=(\ga_1|\ze_1)=(\ze_1|\be_1)=0,\quad
 (\be_1|\be_1)=(\ga_1|\ga_1)=(\ze_1|\ze_1)=2,
\end{equation}
where the bilinear form $(\,|\,)$ on the parameter space is defined by
\begin{equation}\label{eqn:bracket_a}
 (a_i|a_j)=\langle\oc\al_i,\al_j\rangle=A_{ij}.
\end{equation}
\item[(iv)]  {\rm (}Shift actions{\rm )}
Each $\rho_i$\,, $i=0,1,2$, shifts one of $\{\be_1,\ga_1,\ze_1\}$ but not the others, i.e.,
\begin{subequations}
\begin{align}
 &\rho_1(\be_1)-\be_1\in\bbZ_{\neq0},
 &&\rho_1:(\ga_1,\ze_1)\mapsto(\ga_1,\ze_1),\\
 &\rho_2(\ga_1)-\ga_1\in\bbZ_{\neq0},
 &&\rho_2:(\be_1,\ze_1)\mapsto(\be_1,\ze_1),\\
 &\rho_0(\ze_1)-\ze_1\in\bbZ_{\neq0},
 &&\rho_0:(\be_1,\ga_1)\mapsto(\be_1,\ga_1).
\end{align}
\end{subequations}
\end{description}
\item[{\bf II.}] 
Define $z_{l_1,l_2}$ by 
\begin{equation}\label{eqn:def_z_l1l2}
 z_{l_1,l_2}=(-1)^{l_1+l_2}\rho_1^{l_1}\rho_2^{l_2}(z_0),
\end{equation}
where
\begin{equation}\label{eqn:def_z0}
 z_0=\cfrac{{\rho_0}^2(\tau_0)}{\tau_0}.
\end{equation}
Then $z_{l_1,l_2}$ satisfies the cross-ratio equation 
\begin{equation}\label{eqn:cross-ratio}
 \cfrac{(z_{l_1,l_2}-z_{l_1+1,l_2})(z_{l_1+1,l_2+1}-z_{l_1,l_2+1})}{(z_{l_1+1,l_2}-z_{l_1+1,l_2+1})(z_{l_1,l_2+1}-z_{l_1,l_2})}=\cfrac{1}{x}\,,
\end{equation}
and similarity equation \eqref{eqn:similarity}.
\end{enumerate}
\end{theorem}

\begin{remark}
~\\[-1.2em]
\begin{description}
\item[a]
We note that $\bbZ$-module spanned by the parameters $a_i$, $i=0,\dots,4$, gives coroot- or root-lattice of type $\DIV$ denoted by $Q(D_4^{(1)})$, with the pairing $(\,|\,)$.
\item[b]
The properties {\bf (iii)} and {\bf (iv)} are natural requirements since  {\bf (iii)} is the property that $\be_1$, $\ga_1$, $\ze_1$ form simple roots of type $3A_1^{(1)}$ and {\bf (iv)} comes from the idea such that $\rho_i$ are the fundamental translations in $\tW(3A_1^{(1)})$.
\item[c]
The choice of $z_0$ was motivated by the observation that the solution of the cross-ratio equation
is given by a ratio of two $\tau$ functions, where one of the $\tau$ functions lies in a direction orthogonal to the plane in which the equation is iterated \cite{HKM2011:MR2788707}. 
This observation also shows that the variable must be defined in terms of a square of $\rho_0$. 
The remaining directions $\rho_1$ and $\rho_2$ act as the two shift directions.
\item[d] 
By choosing $\tau$-variables appropriately as in \cite{AHKM2014:MR3221837}, we can reproduce the initial values \eqref{eqn:intro_initial0} of the discrete power function solution of \eqref{eqn:intro_cross-ratio} and \eqref{eqn:intro_similarity}. 
In particular, note that the choice $\tau_0=1$ (up to gauge) and ${\rho_0}^2(\tau_0)=0$ leads to the value $z_0=0$.
\item[e]
The property {\bf (i)} is called a {\it projective reduction}, in the sense that we took $\rho_i$, which were not translations in $\tW(\DIV)$, and considered them in a subgroup in which they became translations. 
Many discrete integrable systems can be derived by using such a procedure (see \cite{JNS2015:MR3403054,KN2015:MR3340349,KNT2011:MR2773334,TakenawaT2003:MR1996297}). 
\end{description}
\end{remark}

\begin{proof}
We begin with the construction of the subspace of parameters needed in item {\bf(iii)}. 
Choose 
\begin{equation}
 \be_1=a_2.
\end{equation}
The parameters $\ga_1$ and $\ze_1$ can be chosen from 
\begin{align}
 a_1+a_2+a_4,\quad
 a_1+a_2+a_3,\quad
 a_2+a_3+a_4,
\end{align}
to satisfy the conditions in item {\bf(iii)}. 
For example, take $\ga_1= a_1+a_2+a_4$. 
Then $(\be_1|\ga_1)=0$ and the other required conditions follow from Equations \eqref{eqn:cartan_D4} and \eqref{eqn:bracket_a}.
There is one remaining choice, which we label as $\mu_1$. 
To be specific, from now on, we define 
\begin{equation}
 \ga_1=a_1+a_2+a_4,\quad
 \ze_1=a_1+a_2+a_3,\quad
 \mu_1=a_2+a_3+a_4,
\end{equation}
and also define an additional set of 4 parameters
\begin{equation}
 \be_0=1-\be_1,\quad
 \ga_0=1-\ga_1,\quad
 \ze_0=1-\ze_1,\quad
 \mu_0=1-\mu_1,
\end{equation}
which are needed as simple roots of $4A_1^{(1)}$.

The $\bbZ$-module spanned by the parameters $\be_i$, $\ga_i$, $\ze_i$, $\mu_i$, $i=0,1$, denoted by $Q(4A_1^{(1)})$ gives a coroot- or root-lattice of type $4A_1^{(1)}$, with the pairing $(\,|\,)$.
The transformation group for $Q(4A_1^{(1)})$ forms an extended affine Weyl group of type $4A_1^{(1)}$.  
(See Appendix \ref{section:action_4A1_para} for more detail.)  
Notice that by relinquishing one pair of parameters, say $\mu_i$, $i=0, 1$, we obtain $Q(3A_1^{(1)})$ from $Q(4A_1^{(1)})$.  
We here focus on the sublattice spanned by the parameters $\be_i$, $\ga_i$, $\ze_i$, denoted by $Q(3A_1^{(1)})$.  
Nevertheless, we continue to use the notations including $\mu$ such that
\begin{subequations}\label{eqns:pi_3A1}
\begin{align}
 &\pi_{\be\mu}:(\be_0,\be_1,\mu_0,\mu_1)\mapsto(\be_1,\be_0,\mu_1,\mu_0),\\
 &\pi_{\ga\mu}:(\ga_0,\ga_1,\mu_0,\mu_1)\mapsto(\ga_1,\ga_0,\mu_1,\mu_0),\\
 &\pi_{\ze\mu}:(\ze_0,\ze_1,\mu_0,\mu_1)\mapsto(\ze_1,\ze_0,\mu_1,\mu_0),
\end{align}
\end{subequations}
for consistency with Appendix \ref{section:action_4A1_para}.

We next construct $\rho_i$ in $\tW(3A_1^{(1)})$ by using the elements of $\tW(\DIV)$. 
Recall that $\tW(\DIV) =\langle s_0,\dots,s_4,\si_1,\si_2,\si_3\rangle$. 
The extended affine Weyl group $\tW(3A_1^{(1)})$ corresponding to the parameters chosen above is given by
\begin{equation}
 \widetilde{W}(3A_1^{(1)})=\langle s_{\be_0},s_{\be_1},s_{\ga_0},s_{\ga_1},s_{\ze_0},s_{\ze_1},\pi_{\be\mu},\pi_{\ga\mu},\pi_{\ze\mu}\rangle,
\end{equation}
where
\begin{subequations}\label{eqns:def_W(3A1)}
\begin{align}
 &s_{\be_0}=s_4s_3s_1s_0s_2s_4s_3s_1s_0,\quad
 s_{\be_1}=s_2,\quad
 s_{\ga_0}=s_0s_3s_2s_0s_3,\quad
 s_{\ga_1}=s_1s_4s_2s_1s_4,\\
 &s_{\ze_0}=\si_2 s_1s_3s_2s_1s_3\si_2,\quad
 s_{\ze_1}=s_1s_3s_2s_1s_3,\quad
 \pi_{\be\mu}=\si_1 s_4s_3s_1s_0,\quad
 \pi_{\ga\mu}=\si_3,\\
 &\pi_{\ze\mu}=\si_2.
\end{align}
\end{subequations}
The action of $\tW(3A_1^{(1)})$ on the parameters $\be_i$, $\ga_i$, $\ze_i$, $\mu_i$, $i=0,1$, is given by
\begin{subequations}
\begin{align}
 &s_{\be_0}:(\be_0,\be_1)\mapsto (-\be_0,\be_1+2\be_0),&&
 s_{\be_1}:(\be_0,\be_1)\mapsto (\be_0+2\be_1,-\be_1),\\
 &s_{\ga_0}:(\ga_0,\ga_1)\mapsto (-\ga_0,\ga_1+2\ga_0),&&
 s_{\ga_1}:(\ga_0,\ga_1)\mapsto (\ga_0+2\ga_1,-\ga_1),\\
 &s_{\ze_0}:(\ze_0,\ze_1)\mapsto (-\ze_0,\ze_1+2\ze_0),&&
 s_{\ze_1}:(\ze_0,\ze_1)\mapsto (\ze_0+2\ze_1,-\ze_1),
\end{align}
\end{subequations}
and \eqref{eqns:pi_3A1}.
We are now in a position to define $\rho_i$ by using the translations in $\widetilde{W}(3A_1^{(1)})$:
\begin{equation}\label{eqn:def_rho}
 \rho_1=s_{\be_0}\pi_{\be\mu},\quad
 \rho_2=s_{\ga_0}\pi_{\ga\mu},\quad
 \rho_0=s_{\ze_0}\pi_{\ze\mu},\quad
\end{equation}
whose actions on the parameters $\be_i$, $\ga_i$, $\ze_i$, $\mu_i$, $i=0,1$, are given by
\begin{subequations}
\begin{align}
 &\rho_1:(\be_0,\be_1,\mu_0,\mu_1)\mapsto(\be_0+1,\be_1-1,\mu_1,\mu_0),\\
 &\rho_2:(\ga_0,\ga_1,\mu_0,\mu_1)\mapsto(\ga_0+1,\ga_1-1,\mu_1,\mu_0),\\
 &\rho_0:(\ze_0,\ze_1,\mu_0,\mu_1)\mapsto(\ze_0+1,\ze_1-1,\mu_1,\mu_0).
\end{align}
\end{subequations}
We can easily verify that the transformations $\rho_i$, $i=0,1,2$, satisfy the properties {\bf(i)}, {\bf(ii)} and {\bf(iv)}.

For the second part of the theorem, note that with $z_0$ as defined in \eqref{eqn:def_z0}, the remaining directions $\rho_1$ and $\rho_2$ acts as the two shift directions. 
Now we introduce
\begin{equation}\label{eqn:def_z1z2z12}
 z_1=\rho_1(z_0),\quad
 z_2=\rho_2(z_0),\quad
 z_{12}=\rho_1\rho_2(z_0).
\end{equation}
We can verify that these variables $z_0$, $z_1$, $z_2$ and $z_{12}$ are related by
\begin{equation}\label{eqn:condition_z}
 \cfrac{(z_0+z_1)(z_{12}+z_2)}{(z_1+z_{12})(z_2+z_0)}=\cfrac{1}{x},
\end{equation}
by using the framework constructed, namely, \eqref{eqns:action_D4_tau}, \eqref{eqns:def_W(3A1)}, \eqref{eqn:def_rho}, \eqref{eqn:def_z0} and \eqref{eqn:def_z1z2z12}. 
For example, to express $z_1$ in terms of $\tau$-variables, we need the composition of $\rho_1$ with $\rho_0$.  
The starting point is \eqref{eqn:def_rho} which in turn are given by \eqref{eqns:def_W(3A1)} and the actions on $\tau$ functions \eqref{eqns:action_D4_tau}. 
In this way, we deduce
\begin{subequations}\label{eqns:z0+z1}
\begin{align}
&z_0 + z_1 = \ze_0\dfrac{q}{x^{\frac{1}{2}}}\,\dfrac{\tau_1\tau_3}{\tau_0\tau_4},
&&z_{12} + z_2 = \ze_0\dfrac{x-q}{x^{\frac{1}{2}}(1-q)}\,\dfrac{\tau_1\tau_3}{\tau_0\tau_4},\\
&z_{1} + z_{12} = \ze_0\dfrac{(x-q)q}{x^{\frac{1}{2}}(1-q)}\,\dfrac{\tau_1\tau_3}{\tau_0\tau_4},
&&z_{2} + z_{0} = \ze_0x^{\frac{1}{2}}\,\dfrac{\tau_1\tau_3}{\tau_0\tau_4},
\end{align}
\end{subequations}
where $q=-\ii x^{\frac{1}{2}}\tau_2^{(\si_2)}/\tau_2$. Then we immediately obtain Equation \eqref{eqn:condition_z} from Equations \eqref{eqns:z0+z1}.

Furthermore, from the actions of the inverses of $\rho_1$ and $\rho_2$, we find
\begin{equation}\label{eqn:rho1_2_inv}
 \ze_0 z_0
 =-\be_1\cfrac{(z_1+z_0)\big(z_0+{\rho_1}^{-1}(z_0)\big)}{z_1+{\rho_1}^{-1}(z_0)}
 -\ga_1\cfrac{(z_2+z_0)\big(z_0+{\rho_2}^{-1}(z_0)\big)}{z_2+{\rho_2}^{-1}(z_0)}.
\end{equation}
Using Equations \eqref{eqn:condition_z} and \eqref{eqn:rho1_2_inv} respectively we then obtain the cross-ratio equation \eqref{eqn:cross-ratio} and similarity equation \eqref{eqn:similarity} by defining the dependent variable by \eqref{eqn:def_z_l1l2}.
\end{proof}

\begin{remark}
In the proof of part II of the theorem we used the action of $\tW(D_4^{(1)})$ on $\tau$ functions. 
However, an alternative approach is available by starting with the definitions of the variables $z_0$, $z_1$, $z_2$, $z_{12}$, i.e. \eqref{eqn:def_z0} and \eqref{eqn:def_z1z2z12}. 
For this purpose, action \eqref{eqns:action_D4_tau} is interpreted in terms of the variables $z_0$, $z_1$, $z_2$, $z_{12}$. 
The resulting expressions are very large and to simplify these we can instead introduce $z_{12}$ by using \eqref{eqn:condition_z}, namely,
\begin{equation}\label{eqn:z12_by_z}
 z_{12} = - \frac{xz_2(z_0+z_1)-z_1(z_2+z_0)}{x(z_0+z_1)- (z_2+z_0)}.
\end{equation}
Complete details are provided in Appendix \ref{section:action_3A1_z}.
\end{remark}
\begin{remark}
Let
\begin{equation}
 u_{l_1,l_2,l_0}={\rho_1}^{l_1}{\rho_2}^{l_2}{\rho_0}^{l_0}(z_0),\quad l_1,l_2,l_0\in\bbZ.
\end{equation}
From Equation \eqref{eqn:condition_z} and the following actions of $\rho_i$, $1,2,3$$:$
\begin{subequations}
\begin{align}
 &\left(\frac{1}{z_1}+\frac{1}{z_0}\right) 
 (\rho_0(z_1)+\rho_0(z_0))
 =-\frac{\ze_0(\ze_0+1)}{x},\\
 &\left(\frac{1}{z_2}+\frac{1}{z_0}\right)
 (\rho_0(z_2)+\rho_0(z_0))
 =-\ze_0(\ze_0+1),
\end{align}
\end{subequations}
we obtain
\begin{subequations}\label{eqns:3D_lattice_P6}
\begin{align}
 &\cfrac{(u_{l_1,l_2,l_0}+u_{l_1+1,l_2,l_0})(u_{l_1+1,l_2+1,l_0}+u_{l_1,l_2+1,l_0})}
 {(u_{l_1+1,l_2,l_0}+u_{l_1+1,l_2+1,l_0})(u_{l_1,l_2+1,l_0}+u_{l_1,l_2,l_0})}=\cfrac{1}{x}\,,\label{eqn:face_1}\\
 &\left(\frac{1}{u_{l_1+1,l_2,l_0}}+\frac{1}{u_{l_1,l_2,l_0}}\right) 
 (u_{l_1+1,l_2,l_0+1}+u_{l_1,l_2,l_0+1})
 =-\frac{(\ze_0+l_0) (\ze_0+l_0+1)}{x},\label{eqn:face_2}\\
 &\left(\frac{1}{u_{l_1,l_2+1,l_0}}+\frac{1}{u_{l_1,l_2,l_0}}\right)
 (u_{l_1,l_2+1,l_0+1}+u_{l_1,l_2,l_0+1})
 =-(\ze_0+l_0)(\ze_0+l_0+1).\label{eqn:face_3}
\end{align}
\end{subequations}
These equations are equivalent to the lattice equations of ABS type discussed in \S \ref{section:cubic_lattice}.
Indeed, Equations \eqref{eqns:3D_lattice_P6} can be obtained from Equations \eqref{eqns:3D_lattice_kappa} by the following specialization of parameters:
\begin{equation}
 \ka^{(1)}_l=\dfrac{1}{x},\quad
 \ka^{(2)}_l=1,\quad
 \ka^{(3)}_l=(\ze_0+l)(\ze_0+l+1),
\end{equation}
where $l\in\bbZ$.
\end{remark}

\section{Description by weight lattice}\label{section:weight}
In this section, we relate our results to those obtained from a different perspective, namely the determinantal structure of hypergeometric solutions.  
We show that our $z$-variable \eqref{eqn:def_z_l1l2} is identical to that found in \cite{AHKM2014:MR3221837,HKM2011:MR2788707}.

To make this correspondence, we consider the actions of $\rho_i$, $i=0,1,2$, on the sublattice of the weight lattice. 
We here refer to a vector $(w-{\rm id}).x$ for $w\in\tW(\DIV)$, $x\in M_0$ as a displacement vector corresponding to $w$.  
Since the dependent variable of discrete power function in terms of the $\tau$ function is given by
\begin{equation}\label{eqn:z_tau_1}
 z_{l_1,l_2}=(-1)^{l_1+l_2}\cfrac{{\rho_0}^2{\rho_1}^{l_1}{\rho_2}^{l_2}(\tau_0)}{{\rho_1}^{l_1}{\rho_2}^{l_2}(\tau_0)},
\end{equation}
and $\rho_i$, $i=0,1,2$, satisfy the properties {\bf(i)} and {\bf(ii)} in \S \ref{section:geometric_description},
we here consider the displacement vectors of $\rho_1$ and $\rho_2$ on the following sublattice of $M_0$ \eqref{eqn:def_M_i}:
\begin{equation}
 L=L^{(0)}\sqcup L^{(1)}\sqcup L^{(2)}\sqcup L^{(12)},
\end{equation}
where 
\begin{subequations}
\begin{align}
 &L^{(0)}
 =\left\{{\rho_1}^{2l_1}{\rho_2}^{2l_2}(h_0)\,\Big|\, l_1,l_2\in\bbZ\right\}
 =h_0+\bbZ \bm{v}_1+\bbZ \bm{v}_2,\\
 &L^{(1)}
 =\left\{{\rho_1}^{2l_1+1}{\rho_2}^{2l_2}(h_0)\,\Big|\, l_1,l_2\in\bbZ\right\}
 =-h_1+h_2+h_5+\bbZ \bm{v}_1+\bbZ \bm{v}_2,\\
 &L^{(2)}
 =\left\{{\rho_1}^{2l_1}{\rho_2}^{2l_2+1}(h_0)\,\Big|\, l_1,l_2\in\bbZ\right\}
 =h_4+h_5+h_6+\bbZ \bm{v}_1+\bbZ \bm{v}_2,\\
 &L^{(12)}
 =\left\{{\rho_1}^{2l_1+1}{\rho_2}^{2l_2+1}(h_0)\,\Big|\, l_1,l_2\in\bbZ\right\}
 =h_2-h_3+h_6+\bbZ \bm{v}_1+\bbZ \bm{v}_2.
\end{align}
\end{subequations}
Here, $\bm{v}_i$, $i=0,1,2$, are displacement vectors on $M_0$ given by
\begin{equation}
 \bm{v}_0=-h_0+h_1+h_3-h_4,\quad
 \bm{v}_1=-h_0-h_1+2h_2-h_3-h_4,\quad
 \bm{v}_2=-h_0+h_1-h_3+h_4,
\end{equation}
which correspond to ${\rho_i}^2$, $i=0,1,2$, respectively.
The transformations $\rho_1$ and $\rho_2$ correspond to the vectors on $L$ as the following:
\begin{align}
 &\rho_1
 \leftrightarrow
 \begin{cases}
 -h_0-h_1+h_2+h_5,&\text{on }L^{(0)}\sqcup L^{(12)},\\
 h_2-h_3-h_4-h_5,&\text{on }L^{(1)}\sqcup L^{(2)},
 \end{cases}\\[0.2em]
 &\rho_2
 \leftrightarrow
 \begin{cases}
 -h_0+h_4+h_5+h_6,&\text{on }L^{(0)},\\
 h_1-h_3-h_5+h_6,&\text{on }L^{(1)},\\
 h_1-h_3-h_5-h_6,&\text{on }L^{(2)},\\
 -h_0+h_4+h_5-h_6,&\text{on }L^{(12)}.
 \end{cases}
\end{align}

\begin{figure}[t]
\begin{center}
\includegraphics[width=0.43\textwidth]{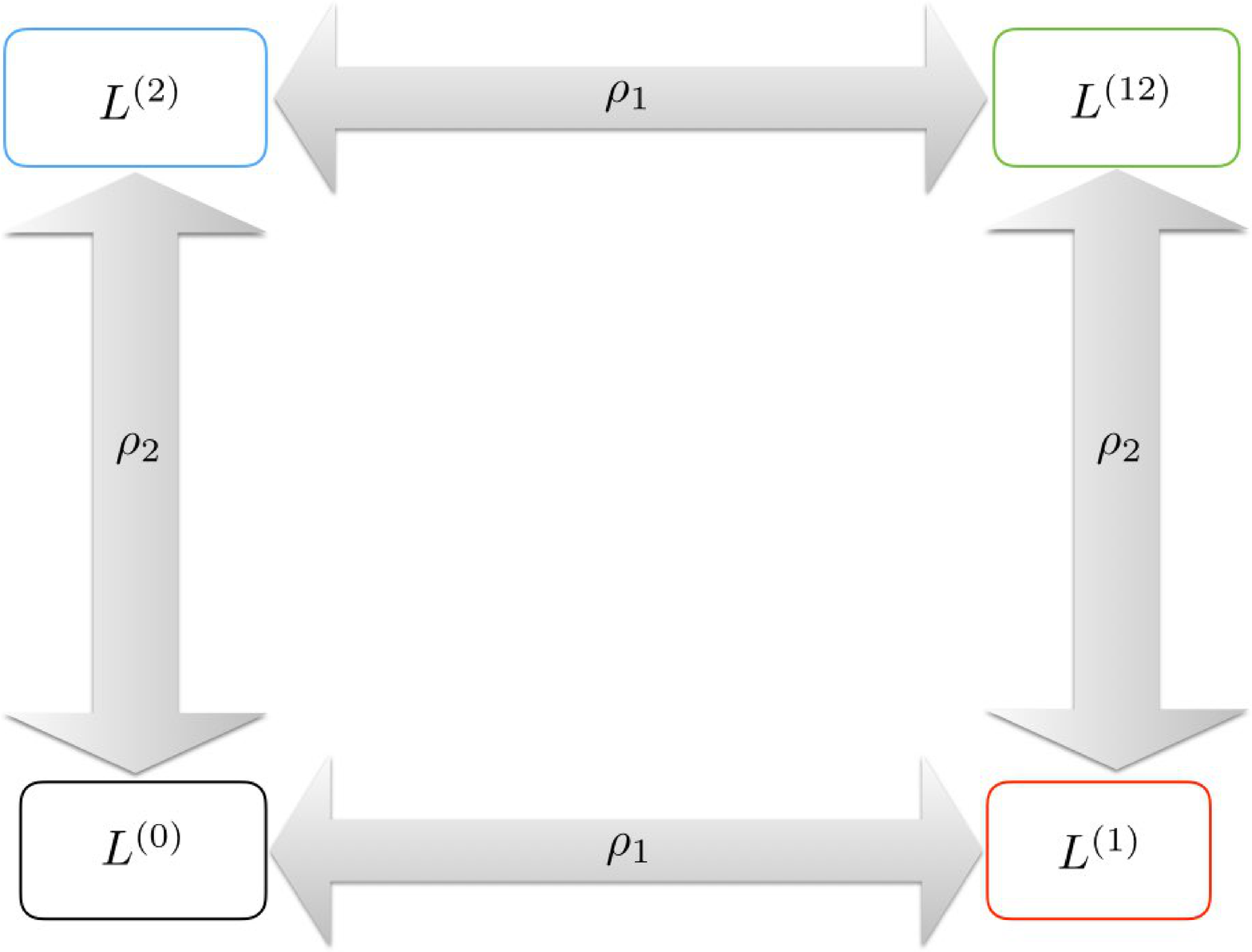}\quad
\includegraphics[width=0.53\textwidth]{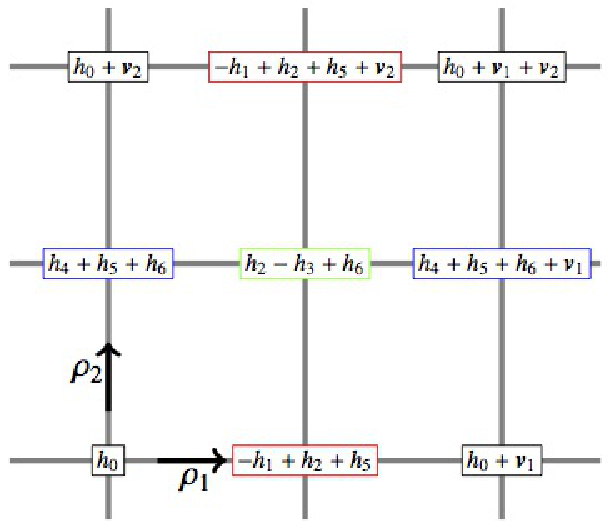}
\end{center}
\caption{The actions of $\rho_1$ and $\rho_2$ on the lattice $L$.
The sub-lattices are distinguished by producing them in different colors, that is,
$L^{(0)}$, $L^{(1)}$, $L^{(2)}$ and $L^{(12)}$ are colored in black, red, blue and green, respectively.}
\label{fig:lattice_L}
\end{figure}

Let us define the following translations in $\tW(\DIV)$ by
\begin{subequations}
\begin{align}
 &\hT_{13}=s_1s_2s_0s_4s_2s_1\si_3,\quad
 \hT_{40}=s_4s_2s_1s_3s_2s_4\si_3,\quad
 \hT_{34}=s_3s_2s_0s_1s_2s_3\si_1,\\
 &T_{14}=s_1s_4s_2s_0s_3s_2\si_2.
\end{align}
\end{subequations}
Using these translations, the actions of ${\rho_0}^2$, $\rho_1$ and $\rho_2$ on $L$ are expressed by
\begin{subequations}
\begin{align}
 &{\rho_0}^2={\hT_{13}}^{~-1}{\hT_{40}}^{~-1}{\hT_{34}}^{~-2},\\
 &\rho_1
 =\begin{cases}
 {\hT_{40}}^{~-1}T_{14},&\text{on }L^{(0)}\sqcup L^{(12)},\\
 {\hT_{13}}^{~-1}T_{14},&\text{on }L^{(1)}\sqcup L^{(2)},
 \end{cases}\\
 &\rho_2
 =\begin{cases}
 {\hT_{40}}^{~-1},&\text{on }L^{(0)}\sqcup L^{(12)},\\
 {\hT_{13}}^{~-1},&\text{on }L^{(1)}\sqcup L^{(2)}.
 \end{cases}
\end{align}
\end{subequations}
Therefore, the dependent variable of discrete power function can be expressed by
\begin{equation}\label{eqn:z_tau_2}
 z_{l_1,l_2}
 =\begin{cases}
 (-1)^{l_1-1}\cfrac{\tau_{-\frac{l_1+l_2}{2}-1,-\frac{l_1+l_2}{2}-1,-2,l_1}}{\tau_{-\frac{l_1+l_2}{2},-\frac{l_1+l_2}{2},0,l_1}}\,,
 &\text{if\, $l_1+l_2$ is even},\\
 (-1)^{l_1-1}\cfrac{\tau_{-\frac{l_1+l_2+1}{2},-\frac{l_1+l_2+1}{2}-1,-2,l_1}}{\tau_{-\frac{l_1+l_2-1}{2},-\frac{l_1+l_2+1}{2},0,l_1}}\,,
 &\text{if\, $l_1+l_2$ is odd},
 \end{cases}
\end{equation}
where
\begin{equation}
 \tau_{k,l,m,n}=\hT_{13}^{~k} {\hT_{40}}^{~l} {\hT_{34}}^{~m} {T_{14}}^{n}.
\end{equation}
This expression is equivalent to that in \cite{AHKM2014:MR3221837,HKM2011:MR2788707}. 
Note that the difference of the coefficients between \eqref{eqn:z_tau_1} and \eqref{eqn:z_tau_2} arises from Equations \eqref{eqns:fundamental_tau_gaps}.
\section{Concluding remarks}\label{ConcludingRemarks}
In this paper we have investigated the geometric structure of the system of partial difference equations \eqref{eqn:intro_cross-ratio} and \eqref{eqn:intro_similarity}.  
The cross-ratio equation \eqref{eqn:intro_cross-ratio} is naturally defined on a multidimensionally consistent cubic lattice as a face equation from the ABS theory.  
On the other hand, we showed in this paper, that its similarity constraint \eqref{eqn:intro_similarity} is derived from the symmetries of the lattice.

However, it was known that the system \eqref{eqn:intro_cross-ratio} and \eqref{eqn:intro_similarity} could be derived from the B\"acklund transformations of P$_{\rm VI}$, which form $\tW(\DIV)$.  
The question of how to relate the symmetry group $\tW(\DIV)$ to the symmetry group of the lattice remained open in the literature.  
We answered the question in this paper by constructing the action of $\tW(3A_1^{(1)})$ as a subgroup of $\tW(\DIV)$.  
Moreover, considering the realization on the level of $\tau$ function simultaneously, we gave the geometric characterizations of the discrete power function.  
In particular, we explained the odd-even structure appearing in the explicit formula in \cite{AHKM2014:MR3221837,HKM2011:MR2788707} in terms of the projective reduction.

It is interesting to relate the geometric structures reported in this paper to other properties of \eqref{eqn:intro_cross-ratio} and \eqref{eqn:intro_similarity} reported in the literature.  
For example, the discrete part of the Lax pair given in \cite{NC1995:MR1329559} can be obtained from translations on the $\tW(3A_1^{(1)})$. 
However, the remaining part of the Lax pair involves a monodromy variable, which is not visible from the point of view of the sixth Painlev\'e equation itself. 
We believe that this is related to the choice of gauge we took in defining the $\tau$ functions. 
This is an open subject for further investigation.

Not many examples of discrete complex analytic functions are known. 
In \cite{AB2003:MR2006759}, one was constructed on the hexagonal lattice.  
Further interesting directions concern the remaining Painlev\'e and discrete Painlev\'e equations.  
We expect our analysis to be useful for such equations, not only of second-order, to give geometric characterization of known examples or to identify new examples.  
Our results in these directions will be reported in forthcoming publications.
\section*{Data accessibility}
This paper does not contain any data.
\section*{Competing interests}
We have no competing interests.
\section*{Authors' contributions}
All authors contributed equally to writing the paper.
\section*{Acknowledgment}
N/A
\section*{Funding statement}
This research was supported by an Australian Laureate Fellowship \# FL120100094 and grant \# DP160101728 from the Australian Research Council
and JSPS KAKENHI Grant Numbers JP16H03941, JP16K13763 and JP17J00092.
\section*{Ethics statement}
This work does not include research on humans or animals.
\appendix
\section{Extended affine Weyl group $\tW(4A_1^{(1)})$}\label{section:action_4A1_para}\allowdisplaybreaks
In this section, we consider an affine Weyl group for $Q(4A_1^{(1)})$ spanned by the parameters $\be_i$, $\ga_i$, $\ze_i$, $\mu_i$, $i=0,1$.

Let
\begin{subequations}\label{eqns:def_W(4A1)}
\begin{align}
 &s_{\be_0}=s_4s_3s_1s_0s_2s_4s_3s_1s_0,\quad
 s_{\be_1}=s_2,\quad
 s_{\ga_0}=s_0s_3s_2s_0s_3,\quad
 s_{\ga_1}=s_1s_4s_2s_1s_4,\\
 &s_{\ze_0}=\si_2 s_1s_3s_2s_1s_3\si_2,\quad
 s_{\ze_1}=s_1s_3s_2s_1s_3,\quad
 s_{\mu_0}=s_0s_1s_2s_0s_1,\\
 &s_{\mu_1}=s_3s_4s_2s_3s_4,\quad
 \pi_{\be\ga}=\si_2 s_4s_3s_1s_0,\quad
 \pi_{\be\ze}=\si_3 s_4s_3s_1s_0,\quad
 \pi_{\be\mu}=\si_1 s_4s_3s_1s_0,\\
 &\pi_{\ga\ze}=\si_1,\quad
 \pi_{\ga\mu}=\si_3,\quad
 \pi_{\ze\mu}=\si_2,\quad
 r_1=s_1s_4,\quad
 r_2=s_1s_3,\quad
 r_3=s_3s_4,
\end{align}
\end{subequations}
whose actions on the parameters $\be_i$, $\ga_i$, $\ze_i$, $\mu_i$, $i=0,1$, are given by
\begin{subequations}
\begin{align}
 &s_{\be_0}:(\be_0,\be_1)\mapsto (-\be_0,\be_1+2\be_0),&&
 s_{\be_1}:(\be_0,\be_1)\mapsto (\be_0+2\be_1,-\be_1),\\
 &s_{\ga_0}:(\ga_0,\ga_1)\mapsto (-\ga_0,\ga_1+2\ga_0),&&
 s_{\ga_1}:(\ga_0,\ga_1)\mapsto (\ga_0+2\ga_1,-\ga_1),\\
 &s_{\ze_0}:(\ze_0,\ze_1)\mapsto (-\ze_0,\ze_1+2\ze_0),&&
 s_{\ze_1}:(\ze_0,\ze_1)\mapsto (\ze_0+2\ze_1,-\ze_1),\\
 &s_{\mu_0}:(\mu_0,\mu_1)\mapsto (-\mu_0,\mu_1+2\mu_0),&&
 s_{\mu_1}:(\mu_0,\mu_1)\mapsto (\mu_0+2\mu_1,-\mu_1),\\
 &\pi_{\be\ga}:(\be_0,\be_1,\ga_0,\ga_1)\mapsto(\be_1,\be_0,\ga_1,\ga_0),&&
 \pi_{\be\ze}:(\be_0,\be_1,\ze_0,\ze_1)\mapsto(\be_1,\be_0,\ze_1,\ze_0),\\
 &\pi_{\be\mu}:(\be_0,\be_1,\mu_0,\mu_1)\mapsto(\be_1,\be_0,\mu_1,\mu_0),&&
 \pi_{\ga\ze}:(\ga_0,\ga_1,\ze_0,\ze_1)\mapsto(\ga_1,\ga_0,\ze_1,\ze_0),\\
 &\pi_{\ga\mu}:(\ga_0,\ga_1,\mu_0,\mu_1)\mapsto(\ga_1,\ga_0,\mu_1,\mu_0),&&
 \pi_{\ze\mu}:(\ze_0,\ze_1,\mu_0,\mu_1)\mapsto(\ze_1,\ze_0,\mu_1,\mu_0),\\
 &r_1:(\be_i,\ga_i,\ze_i,\mu_i)\mapsto(\ga_i,\be_i,\mu_i,\ze_i),&&
 r_2:(\be_i,\ga_i,\ze_i,\mu_i)\mapsto(\ze_i,\mu_i,\be_i,\ga_i),\\
 &r_3:(\be_i,\ga_i,\ze_i,\mu_i)\mapsto(\mu_i,\ze_i,\ga_i,\be_i),
\end{align}
\end{subequations}
where $i=0,1$.
They collectively form the extended affine Weyl group of type $4A_1^{(1)}$:
\begin{equation}
 \tW(4A_1^{(1)})
 =\langle s_{\be_0},s_{\be_1},s_{\ga_0},s_{\ga_1},s_{\ze_0},s_{\ze_1},s_{\mu_0},s_{\mu_1},\pi_{\be\ga},\pi_{\be\ze},\pi_{\be\mu},\pi_{\ga\ze},\pi_{\ga\mu},\pi_{\ze\mu},r_1,r_2,r_3\rangle.
\end{equation}
Indeed, $\tW(4A_1^{(1)})$ satisfies the following fundamental relations:
\begin{subequations}\label{eqns:fundamental_rels_A1A1A1A1}
\begin{align}
 &{s_{\be_0}}^2={s_{\be_1}}^2={s_{\ga_0}}^2={s_{\ga_1}}^2={s_{\ze_0}}^2={s_{\ze_1}}^2={s_{\mu_0}}^2={s_{\mu_1}}^2=1,\\
 &(s_{\be_0}s_{\be_1})^\infty=(s_{\ga_0}s_{\ga_1})^\infty=(s_{\ze_0}s_{\ze_1})^\infty=(s_{\mu_0}s_{\mu_1})^\infty=1,\\
 &s_{\be_i}s_{\ga_j}=s_{\ga_i}s_{\be_j},\quad
 s_{\be_i}s_{\ze_j}=s_{\ze_i}s_{\be_j},\quad
 s_{\be_i}s_{\mu_j}=s_{\mu_i}s_{\be_j},\quad
 s_{\ga_i}s_{\ze_j}=s_{\ze_i}s_{\ga_j},\\
 &s_{\ga_i}s_{\mu_j}=s_{\mu_i}s_{\ga_j},\quad
 s_{\ze_i}s_{\mu_j}=s_{\mu_i}s_{\ze_j},\\
 &{\pi_{\be\ga}}^2={\pi_{\be\ze}}^2={\pi_{\be\mu}}^2={\pi_{\ga\ze}}^2={\pi_{\ga\mu}}^2={\pi_{\ze\mu}}^2=1,\\
 &\pi_{\ga\ze}=\pi_{\be\ga}\pi_{\be\ze},\quad
 \pi_{\ga\mu}=\pi_{\be\ga}\pi_{\be\mu},\quad
 \pi_{\ze\mu}=\pi_{\be\ze}\pi_{\be\mu},\\
 &\pi_{\be\ga}\pi_{\be\ze}=\pi_{\be\ze}\pi_{\be\ga},\quad
 \pi_{\be\ga}\pi_{\be\mu}=\pi_{\be\mu}\pi_{\be\ga},\quad
 \pi_{\be\ze}\pi_{\be\mu}=\pi_{\be\mu}\pi_{\be\ze},\\
 &{r_1}^2={r_2}^2={r_3}^2=1,\quad
 r_3=r_1r_2=r_2r_1,
\end{align}
\begin{align}
 &\pi_{\be\ga}s_{\{\be_0,\be_1,\ga_0,\ga_1,\ze_i,\mu_i\}}
 =s_{\{\be_1,\be_0,\ga_1,\ga_0,\ze_i,\mu_i\}}\pi_{\be\ga},\quad
 \pi_{\be\ze}s_{\{\be_0,\be_1,\ga_i,\ze_0,\ze_1,\mu_i\}}
 =s_{\{\be_1,\be_0,\ga_i,\ze_1,\ze_0,\mu_i\}}\pi_{\be\ze},\\
 &\pi_{\be\mu}s_{\{\be_0,\be_1,\ga_i,\ze_i,\mu_0,\mu_1\}}
 =s_{\{\be_1,\be_0,\ga_i,\ze_i,\mu_1,\mu_0\}}\pi_{\be\mu},\\
 &r_1s_{\{\be_i,\ga_i,\ze_i,\mu_i\}}=s_{\{\ga_i,\be_i,\mu_i,\ze_i\}}r_1,\quad
 r_2s_{\{\be_i,\ga_i,\ze_i,\mu_i\}}=s_{\{\ze_i,\mu_i,\be_i,\ga_i\}}r_2,\\
 &r_1\pi_{\{\be\ga,\be\ze,\be\mu,\ga\ze,\ga\mu,\ze\mu\}}=\pi_{\{\be\ga,\ga\mu,\ga\ze,\be\mu,\be\ze,\ze\mu\}}r_1,\quad
 r_2\pi_{\{\be\ga,\be\ze,\be\mu,\ga\ze,\ga\mu,\ze\mu\}}=\pi_{\{\ze\mu,\be\ze,\ga\ze,\be\mu,\ga\mu,\be\ga\}}r_2.
\end{align}
\end{subequations}
Here, the relation $(ww')^\infty=1$ for transformations $w$ and $w'$ means that there is no positive integer $N$ such that $(ww')^N=1$.
Note that from the definitions \eqref{eqns:def_W(3A1)} and \eqref{eqns:def_W(4A1)}, $\tW(3A_1^{(1)})\subset \tW(4A_1^{(1)})$ holds.
\section{Birational actions of $\tW(3A_1^{(1)})$ on the $z$-variables}\label{section:action_3A1_z}
The action of $\tW(3A_1^{(1)})$ on the variables $z_0$, $z_1$, $z_2$, $z_{12}$, where $z_{12}$ is related to other variables as \eqref{eqn:z12_by_z}, is given by the following:
\begin{subequations}\label{eqns:action_z_3A1}
\begin{align}
 &s_{\be_0}(z_0)
 =\frac{2 \be_0 z_1z_{12}-(\be_0+\ga_0+\ze_0-\mu_0)z_0z_1+(\be_0-\ga_0+\ze_0+\mu_0)z_0z_{12}}
 {2 \be_0 z_0+(\be_0-\ga_0-\ze_0+\mu_0)z_1-(\be_0+\ga_0-\ze_0-\mu_0)z_{12}},\\
 &s_{\be_0}(z_2)
 =\frac{2 \be_0 z_1z_{12}+(\be_0+\ga_0+\ze_0-\mu_0)z_1z_2-(\be_0-\ga_0+\ze_0+\mu_0)z_2z_{12}}
 {2 \be_0 z_2-(\be_0-\ga_0-\ze_0+\mu_0)z_1+(\be_0+\ga_0-\ze_0-\mu_0)z_{12}},\\
 &s_{\ga_0}(z_0)
 =\frac{2 \ga_0 z_2z_{12}-(\be_0-\ga_0-\ze_0-\mu_0)z_0z_{12}-(\be_0+\ga_0+\ze_0-\mu_0)z_0z_2}
 {2 \ga_0 z_0-(\be_0+\ga_0-\ze_0-\mu_0)z_{12}-(\be_0-\ga_0+\ze_0-\mu_0)z_2},\\
 &s_{\ga_0}(z_1)
 =\frac{2 \ga_0 z_2z_{12}+(\be_0-\ga_0-\ze_0-\mu_0)z_1z_{12}+(\be_0+\ga_0+\ze_0-\mu_0)z_1z_2}
 {2 \ga_0 z_1+(\be_0+\ga_0-\ze_0-\mu_0)z_{12}+(\be_0-\ga_0+\ze_0-\mu_0)z_2},\\
 &s_{\ze_0}:(z_0,z_1,z_2,z_{12})\mapsto(-{z_0}^{-1},-{z_1}^{-1},-{z_2}^{-1},-{z_{12}}^{-1}),\\
 &s_{\be_1}=\pi_{\be\mu}s_{\be_0}\pi_{\be\mu},\quad
 s_{\ga_1}=\pi_{\ga\mu}s_{\ga_0}\pi_{\ga\mu},\quad
 s_{\ze_1}=\pi_{\ze\mu}s_{\ze_0}\pi_{\ze\mu},\\
 &\pi_{\be\mu}:(z_0,z_1,z_2,z_{12})\mapsto(z_1,z_0,z_{12},z_2),\quad
 \pi_{\ga\mu}:(z_0,z_1,z_2,z_{12})\mapsto(z_2,z_{12},z_0,z_1),
\end{align}
\begin{align}
 \pi_{\ze\mu}(z_0)
 =&-\ze_0\frac{(\be_0-\ga_1+\ze_0-\mu_1)z_0+(\be_0+\ga_1-\ze_0-\mu_1)z_1}{2 (z_0+z_2) (z_1+z_{12})}\notag\\
 &-\ze_0\frac{(\be_1-\ga_1-\ze_0+\mu_0)z_2+(\be_1+\ga_1-\ze_0+\mu_0)z_{12}}{2 (z_0+z_2) (z_1+z_{12})},\\
 \pi_{\ze\mu}(z_1)
 =&-\ze_0\frac{(\be_1+\ga_1-\ze_0-\mu_0)z_0+(\be_1-\ga_1+\ze_0-\mu_0)z_1}{2 (z_0+z_2) (z_1+z_{12})}\notag\\
 &-\ze_0\frac{(\be_0+\ga_1-\ze_0+\mu_1)z_2+(\be_0-\ga_1-\ze_0+\mu_1)z_{12}}{2 (z_0+z_2) (z_1+z_{12})},\\
 \pi_{\ze\mu}(z_2)
 =&-\ze_0\frac{(\be_1-\ga_0-\ze_0+\mu_1)z_0+(\be_1+\ga_0-\ze_0+\mu_1)z_1}{2 (z_0+z_2) (z_1+z_{12})}\notag\\
 &-\ze_0\frac{(\be_0-\ga_0+\ze_0-\mu_0)z_2+(\be_0+\ga_0-\ze_0-\mu_0)z_{12}}{2 (z_0+z_2) (z_1+z_{12})},\\
 \pi_{\ze\mu}(z_{12})
 =&-\ze_0\frac{(\be_0+\ga_0-\ze_0+\mu_0)z_0+(\be_0-\ga_0-\ze_0+\mu_0)z_1}{2 (z_0+z_2) (z_1+z_{12})}\notag\\
 &-\ze_0\frac{(\be_1+\ga_0-\ze_0-\mu_1)z_2+(\be_1-\ga_0+\ze_0-\mu_1)z_{12}}{2 (z_0+z_2) (z_1+z_{12})}.
\end{align}
\end{subequations}
Under the actions on the parameters \eqref{eqns:pi_3A1} and \eqref{eqns:def_W(3A1)} and $z$-variables \eqref{eqns:action_z_3A1}, the following fundamental relations hold:
\begin{subequations}\label{eqns:fundamental_rels_3A1}
\begin{align}
 &{s_{\be_i}}^2={s_{\ga_i}}^2={s_{\ze_i}}^2=1,\quad i=0,1,\quad
 (s_{\be_0}s_{\be_1})^\infty=(s_{\ga_0}s_{\ga_1})^\infty=(s_{\ze_0}s_{\ze_1})^\infty=1,\\
 &(s_{\be_i}s_{\ga_j})^2=(s_{\be_i}s_{\ze_j})^2=(s_{\ga_i}s_{\ze_j})^2=1,\quad i,j=0,1,\quad
 {\pi_{\be\mu}}^2={\pi_{\ga\mu}}^2={\pi_{\ze\mu}}^2=1,\\
 &\pi_{\be\mu}s_{\{\be_k,\ga_k,\ze_k\}}
 =s_{\{\be_{k+1}\ga_k,\ze_k\}}\pi_{\be\mu},\quad
 \pi_{\ga\mu}s_{\{\be_k,\ga_k,\ze_k\}}
 =s_{\{\be_k,\ga_{k+1},\ze_k\}}\pi_{\ga\mu},\\
 &\pi_{\ze\mu}s_{\{\be_k,\ga_k,\ze_k\}}
 =s_{\{\be_k,\ga_k,\ze_{k+1}\}}\pi_{\ze\mu},\quad k\in\bbZ/(2\bbZ).
\end{align}
\end{subequations}
\def\cprime{$'$} \def\cprime{$'$}

\end{document}